\documentclass[runningheads]{llncs}
\usepackage{graphicx}
\usepackage{stmaryrd}
\usepackage{amsmath}
\usepackage{amssymb}
\usepackage{subcaption}
\newtheorem{fact}[theorem]{Fact}
\usepackage{tikz} 
\newcommand\trans[3]{#1 \xrightarrow[]{#2}\hspace{-.3em}\mathstrut^\ast \hspace{.2em} #3}
\newcommand\transo[3]{#1 \xrightarrow[]{#2}\hspace{-.3em}\mathstrut^\omega \hspace{.2em} #3}

\usepackage{enumitem, mathabx}
\renewcommand\problem[2]{\begin{description}[topsep=0pt]\item[Input:] #1\item[Question:]#2\end{description}}
\begin{document}
\title{Sturmian and infinitely desubstitutable words accepted by an $\omega$-automaton}
\titlerunning{Sturmian and infinitely desubstitutable words accepted by an $\omega$-automaton}
%
\author{Pierre Béaur\inst{1}\orcidID{0000-0002-4050-7709} \and
Benjamin Hellouin de Menibus\inst{1}\orcidID{0000-0001-5194-929X}}
\authorrunning{P. Béaur and B. Hellouin de Menibus}
%
\institute{Université Paris-Saclay, CNRS, Laboratoire Interdisciplinaire des Sciences du Numérique, 91400, Orsay, France}
\maketitle              
\begin{abstract}
Given an $\omega$-automaton and a set of substitutions, we look at which accepted words can also be defined through these substitutions, and in particular if there is at least one.
We introduce a method using desubstitution of $\omega$-automata to describe the structure of preimages of accepted words under arbitrary sequences of homomorphisms: this takes the form of a meta-$\omega$-automaton.

We decide the existence of an accepted purely substitutive word, as well as the existence of an accepted fixed point.
In the case of multiple substitutions (non-erasing homomorphisms), we decide the existence of an accepted infinitely desubstitutable word, with possibly some constraints on the sequence of substitutions (\emph{e.g.} Sturmian words or Arnoux-Rauzy words).
As an application, we decide when a set of finite words codes \emph{e.g.} a Sturmian word.
As another application, we also show that if an $\omega$-automaton accepts a Sturmian word, it accepts the image of the full shift under some Sturmian morphism.

\end{abstract}

\section{Introduction}

One-dimensional symbolic dynamics is the study of infinite words and their associated dynamical structures, and is linked with combinatorics on words.
Two classical methods to generate words are the following: on the one hand, sofic shifts are the set of infinite walks on a labeled graph (which can be considered as an $\omega$-automaton) \cite{lind_marcus_1995}; 
on the other hand, the substitutive approach consists in iterating a word homomorphism on an initial letter. 
The latter method was introduced by Axel Thue as a way to create counterexamples to conjectures in combinatorics on words \cite{berstel_1992}.

These two constructions usually build words and languages which are of a very different nature.
On the one hand, substitutive words tend to have a self-similar structure, and are used to generate minimal aperiodic subshifts ; on
the other hand, sofic shifts always contain ultimately periodic words and cannot be minimal if they contain a non-periodic word.
We aim at deciding when a given $\omega$-automaton accepts a word with a given substitutive structure, and study the properties of sets of such accepted words.
Carton and Thomas provided a method to decide this question in the case of substitutive or morphic words on Büchi $\omega$-automata, using verification theory and semigroups of congruence \cite{carton_thomas_2002}.
This result was partially reproved by Salo \cite{salo_2022}, using a more combinatorial point of view.
For the last 20 years, the substitutive approach (iterating a single homomorphism) has been generalized to the S-adic approach \cite{ferenczi_1996} that lets one alternate betweeen multiple substitutions.
This more general framework lets us describe other natural classes, such as the family of Sturmian words.

In this paper, we develop a new method based on desubstitutions of $\omega$-automata.
We can express the preimages of an $\omega$-automaton by any sequence of substitutions through a meta-$\omega$-automaton, whose vertices are $\omega$-automata and whose edges are labeled by substitutions.
We use this meta-$\omega$-automaton to decide whether an $\omega$-automaton accepts a purely substitutive word (giving an alternative proof of \cite{carton_thomas_2002}), or a fixed point of a substitution, or a morphic word, or an infinitely desubstitutable word (by a set of substitutions).
The method is flexible enough to enforce additional constraints on the directive sequences of substitutions, which is powerful enough for example to decide whether an $\omega$-automaton accepts a Sturmian word.
A consequence is the decidability of whether a given set of finite words codes some Sturmian word (or from any family of words with an $S$-adic characterization).
We also describe the set of directive sequences of words accepted by some $\omega$-automaton, which is an $\omega$-regular set.

The meta-$\omega$-automaton also provides a more combinatorial insight on how Sturmian words and $\omega$-regular languages interact: namely, that an $\omega$-automaton accepts a Sturmian word if, and only if, it accepts the image of the full shift under a Sturmian morphism.

\section{Definitions}
\subsection{Words and $\omega$-automata}

An alphabet $\mathcal{A}$ is a finite set of symbols.
The set of finite words on $\mathcal{A}$ is denoted as $\mathcal{A}^*$, and contains the empty word.
A (mono)infinite word  is an element of $\mathcal{A}^{\mathbb{N}}$.
It is usual to write $x = x_0 x_1 x_2 x_3 \dots$ where $x_i = x(i) \in \mathcal{A}$.
If $x$ is a word, $|x|$ is the length of the word (if $x$ is infinite, then $|x| = \infty$).
For a word $x$ and $0 \leq j \leq k < |x|$, $x_{\llbracket j,k \rrbracket}$ is the word $x_j x_{j+1} x_{j+2} \dots x_{k-1} x_{k}$.
We denote $w \sqsubseteq_p x$ when $w$ is a prefix of $x$, that is, $w = x_{\llbracket 0,k \rrbracket}$.

It is possible to endow $\mathcal{A}^\mathbb{N}$ with a topology, called \emph{the prodiscrete topology}.
The prodiscrete topology is defined by the clopen basis $[w]_n = \{x \in \mathcal{A}^\mathbb{N}\ |\ x_n x_{n+1} \dots x_{n+|w|-1} = w\}$ for $w \in \mathcal{A}^\ast$.
To this topology, we can adjunct a dynamic with the shift operator $S$:
\[S : \left(\begin{array}{ccc}
    \mathcal{A}^\mathbb{N} & \rightarrow & \mathcal{A}^\mathbb{N} \\
    x = x_0 x_1 x_2 x_3 \dots & \mapsto & S(x) = x_1 x_2 x_3 x_4 \dots 
\end{array}\right)\]
A set $X \subseteq \mathcal{A}^\mathbb{N}$ is called \emph{a shift (space)} if it is stable by $S$ and closed for the prodiscrete topology.
In particular, $\mathcal{A}^\mathbb{N}$ is a shift space, called \emph{the full shift (space)}.

We now introduce the main computational model of this paper: $\omega$-automata.

\begin{definition}[$\omega$-automaton]
An $\omega$-automaton $\mathfrak{A}$ is a tuple $(\mathcal{A}, Q, I, T
)$, where $\mathcal{A}$ is an alphabet, $Q$ is a finite set of states, $I \subseteq Q$ is the set of initial states, $T \subseteq Q \times \mathcal{A} \times Q$ is the set of transitions of $\mathfrak{A}$.
\end{definition}

We extend several classical notions from finite automata.
We write transitions as $q_s \xrightarrow[]{a} q_t \in T$.

\begin{definition}[Computations and walks]
For $n \geq 1$ or $n = \infty$, a sequence $(q_k)_{0 \leq k \leq n}$ with $q_k \in Q$ is a \emph{walk} in $\mathfrak{A}$ if there is $(a_k)_{1\leq k\leq n} \subseteq \mathcal{A}$ such that for all $0 \leq k \leq n-1$, $q_k \xrightarrow[]{a_{k+1}} q_{k+1}\in T$. We then write $q_0 \xrightarrow[]{a_1} q_1 \xrightarrow[]{a_2} q_2 \xrightarrow[]{a_3} \cdots \xrightarrow[]{a_n} q_{n}$.
The word $w = (a_k)_{1 \leq k \leq n}$ \emph{labels} the  walk, and we call \emph{computation} a labeled walk.
If the computation begins with an initial state, $w$ is \emph{accepted} by $\mathfrak A$.


\end{definition}

In the literature, $\omega$-automata usually have an acceptance condition (such as the Büchi condition \cite{thomas_1997}).
In this paper, we will consider $\omega$-automata to have the largest acceptance condition: every walk beginning with an initial state is accepting.
This is a weaker model than Büchi $\omega$-automata.


\begin{definition}[Language of an $\omega$-automaton]
Let $\mathfrak{A}$ be an $\omega$-automaton.
The language of \textbf{finite} words of $\mathfrak{A}$ is $\mathcal{L}_F(\mathfrak{A}) = \{ w \in \mathcal{A}^*\ |\ w \text{ is accepted by } \mathfrak{A} \}$. The language of \textbf{infinite} words of $\mathfrak{A}$ is  $\mathcal{L}_\infty(\mathfrak{A})) \{ w \in \mathcal{A}^\mathbb{N} \ |\ w \text{ is accepted by } \mathfrak{A}\}$.
Then, the language of $\mathfrak{A}$ is $\mathcal{L}(\mathfrak{A}) = \mathcal{L}_F(\mathfrak{A}) \cup \mathcal{L}_\infty(\mathfrak{A})$.
\end{definition}

If all states of $\mathfrak A$ are initial ($I = Q$), its language of infinite words is a shift, called a \emph{sofic shift} \cite{lind_marcus_1995}.

\subsection{Substitutions}
\begin{definition}[Homomorphisms and substitutions]
A homomorphism is a function $\sigma : \mathcal{A}^\ast \rightarrow \mathcal{A}^\ast$ such that $\sigma(uv) = \sigma(u)\sigma(v)$ (concatenation) for all $u,v \in \mathcal{A}^\ast$.
The homomorphism $\sigma$ is extended to $\mathcal{A}^\mathbb{N} \rightarrow \mathcal{A}^\mathbb{N}$ by $\sigma(x_0 x_1 x_2 \dots) = \sigma(x_0) \sigma(x_1) \sigma(x_2) \dots$
A substitution is a \emph{nonerasing} homomorphism, that is, $\sigma(a) \neq \varepsilon$ for all letters $a \in \mathcal{A}$.
\end{definition}

\begin{definition}[Fixed points, purely substitutive, substitutive and morphic words]
Let $\sigma, \tau : \mathcal{A}^\mathbb{N} \rightarrow \mathcal{A}^\mathbb{N}$ be two homomorphisms.
An infinite word $x \in \mathcal{A}^\mathbb{N}$ is:
\begin{itemize}
    \item a \emph{fixed point} of $\sigma$ if $\sigma(x) = x$;
    \item a \emph{purely substitutive word} generated by $\sigma$ if there is a letter $a \in \mathcal{A}$ such that $x = \lim\limits_{n \rightarrow \infty} \sigma^n(a)$, where the limit is well-defined;
    \item a \emph{morphic word} generated by $\sigma$ and $\tau$ if $x = \tau(y)$, where $y$ is a purely substitutive word generated by $\sigma$;
    \item a \emph{substitutive word} generated by $\sigma$ if $x$ is a morphic word generated by $\sigma$ and a coding $\tau$, i.e. $\tau(\mathcal{A}) \subseteq \mathcal{A}$.
\end{itemize}
\end{definition}

It is now possible to extend these definitions to the case where we use multiple homomorphisms.
However, most of literature revolves around the use of multiple non-erasing homomorphisms (substitutions), and we will stick to this case.
Let $(\sigma_n)_{n \in \mathbb N}$ be a sequence of substitutions.
The equivalent of a fixed-point of one homomorphism is an \emph{infinitely desubstitutable word} by a sequence of substitutions:

\begin{definition}[Infinitely desubstitutable words and directive sequences]
    Let $\mathcal{S}$ be a finite set of substitutions on a single alphabet $\mathcal{A}$, and let $(\sigma_n)_{n \in \mathbb N} \subseteq \mathcal{S}$.
    An infinite word $x$ is \emph{infinitely desubstitutable} by $(\sigma_n)_{n \in \mathbb N}$ (called a \emph{directive sequence} of $x$) if, and only if, there exists a sequence of infinite words $(x_n)_{n \in \mathbb N}$ such that $x_0 = x$ and $x_{n} = \sigma_n(x_{n+1})$.
    An infinite word $x$ is infinitely desubstitutable by $\mathcal{S}$ if $x$ is infinitely desubstitutable by some directive sequence $(\sigma_n)_{n \in \mathbb N} \subseteq \mathcal{S}$.
\end{definition}

Just like for words, we write $\sigma_{\llbracket i,j \rrbracket} = \sigma_i \circ \sigma_{i+1} \circ \dots \circ \sigma_j$.
Then, by compactness of $\mathcal{A}^{\mathbb N}$, $x$ is infinitely desubstitutable by $(\sigma_n)_{n \in \mathbb N}$ if, and only if, there is a sequence of infinite words $(x_n)_{n \in \mathbb N}$ such that $x = \sigma_{\llbracket 0, n \rrbracket}(x_{n+1}) $ for all $n \geq 0$.

\section{Finding substitutive and infinitely desubstitutable words in $\omega$-automata}

\subsection{Desubstituting $\omega$-automata}

In this section, we explain our main technical tool: an effective transformation of $\omega$-automaton, called desubstitution.
We define it for the broad case of possibly erasing homomorphisms.

\begin{definition}[Desubstitution of an $\omega$-automaton]\label{def:DesubAuto}
Let $\mathfrak{A} = (\mathcal{A}, Q, I, T)$ be an $\omega$-automaton, and $\sigma$ a homomorphism.
We define $\sigma^{-1}(\mathfrak{A})$ as the $\omega$-automaton $(\mathcal{A}, Q, I, T')$ where, for all $q_1,q_2 \in Q$ and $a\in\mathcal A$, $q_1 \xrightarrow[]{a} q_2 \in T'$ iff $\trans{q_1}{\sigma(a)}{q_2}$ is a computation in $\mathfrak{A}$.
\end{definition}

In particular, in this case, we consider that $q \xrightarrow{\varepsilon} q$ is a computation.
Thus, if $\sigma(a) = \varepsilon$, the desubstituted automaton $\sigma^{-1}(\mathfrak{A})$ has a loop labeled by $a$ on every state.

For example, consider the following $\omega$-automaton $\mathfrak{A}$ and substitution $\sigma$ (Figure~\ref{fig:desubst}(a,b)).

\begin{figure}[ht!]
    \centering
    \begin{subfigure}[b]{0.2\textwidth}
        \centering
        $\arraycolsep=2pt\def\arraystretch{1}
\sigma : \left\{\begin{array}{ccc}
    0 & \mapsto & 01  \\
    1 & \mapsto & 0
\end{array}\right.$

        \vspace{0.4\textwidth}
        \caption{$\sigma$}
    \end{subfigure}
    \hfill    
    \begin{subfigure}[b]{0.2\textwidth}
         \centering
         \includegraphics[width = \textwidth]{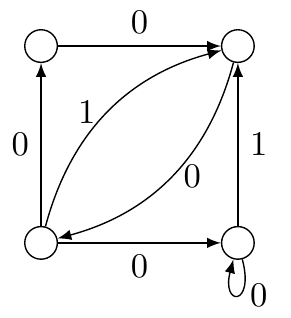}
         \caption{$\mathfrak{A}$}
     \end{subfigure}
     \hfill
    \begin{subfigure}[b]{0.25\textwidth}
         \centering
         \includegraphics[width = 0.8\textwidth]{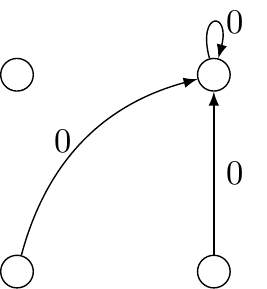}
         \caption{Intermediate step}
     \end{subfigure}
     \hfill
    \begin{subfigure}[b]{0.2\textwidth}
         \centering
         \includegraphics[width = \textwidth]{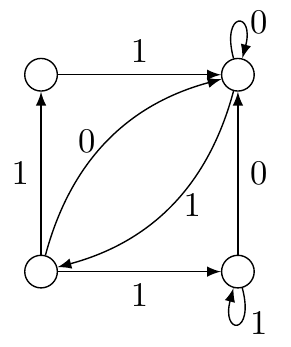}
         \caption{$\sigma^{-1}(\mathfrak{A})$}
     \end{subfigure}
    
    \caption{Desubstitution of the $\omega$-automaton $\mathfrak{A}$ by $\sigma$}\label{fig:desubst}
\end{figure}

We build the $\omega$-automaton $\sigma^{-1}(\mathfrak{A})$.
Start from an empty automaton on the same set of states. 
For every computation in $\mathfrak{A}$ labeled by $01 = \sigma(0)$ --- say, $\trans{q}{\sigma(0)}{r}$ --- add an edge $q\xrightarrow[]{0} r$ to the automaton (Figure~\ref{fig:desubst}(c)).
To conclude, do this with $\sigma(1) = 0$ (Figure~\ref{fig:desubst}(d)).

Stability by inverse morphism is a classical concept in the theory of finite automata \cite{hopcroft_ullman_1979}, and desubstitution satisfies the following property:

\begin{proposition}\label{prop:DesubAuto}
An infinite word $u$ is accepted by $\sigma^{-1}(\mathfrak{A})$ if and only if $\sigma(u)$ is accepted by $\mathfrak{A}$. In other words, $\mathcal L_\infty(\sigma^{-1}(\mathfrak{A})) = \sigma^{-1}(\mathcal L_\infty(\mathfrak{A}))$.
\end{proposition}

\begin{proof}
Let $u$ be accepted by $\sigma^{-1}(\mathfrak{A})$.
Consider the associated accepting walk $(q_i)_{i \in \mathbb{N}}$.
By definition of $\sigma^{-1}(\mathfrak{A})$, for every $i \in \mathbb{N}$, there exists a computation $\trans{q_i}{\sigma(u_i)}{q_{i+1}}$ in $\mathfrak{A}$.
By concatenating these computations, we get an infinite computation $\trans{q_0}{\sigma(u_0)}{q_1} \trans{}{\sigma(u_1)}{q_2} \trans{}{\sigma(u_2)}{\cdots} $ in $\mathfrak{A}$ that accepts $\sigma(u)$ in $\mathfrak{A}$.

Conversely, suppose there is a word of the form $\sigma(u)$ accepted by $\mathfrak{A}$.
Consider the states $(q_i)_{i \in \mathbb{N}}$ obtained after reading each $\sigma(a)$ for $a \in \mathcal{A}$.
This defines an accepting computation labeled by $u$ in $\sigma^{-1}(\mathfrak{A})$.
\end{proof}

This proof actuallly provides a similar result for finite words:

\begin{proposition}
\label{WordSub}
    Let $w$ be a finite word, $\mathfrak{A}$ an $\omega$-automaton and $\sigma$ a homomorphism.
    Then $\trans{q_s}{\sigma(w)}{q_t}$ is a computation in $\mathfrak{A}$ iff $\trans{q_s}{w}{q_t}$ is a computation in $\sigma^{-1}(\mathfrak{A})$.
\end{proposition}

An easy but significant property is the composition of desubstitution of $\omega$-automata:

\begin{proposition}\label{prop:DesubComp}
Let $\mathfrak{A}$ be an $\omega$-automaton, and $\sigma$ and $\tau$ be two homomorphisms.
Then, $(\sigma \circ \tau)^{-1}(\mathfrak{A}) = \tau^{-1} (\sigma^{-1}(\mathfrak{A}))$.
\end{proposition}

\begin{proof}
These two $\omega$-automata share the same sets of states and of initial states.
We prove that they have the same transitions.
We have indeed:
\begin{align*}
        q_s \xrightarrow[]{a} q_t \text{ in } (\sigma \circ \tau)^{-1}(\mathfrak{A}) &\iff \trans{q_s}{\sigma \circ \tau (a)}{q_t} \text{ in } \mathfrak{A} \\
        &\iff \trans{q_s}{\tau(a)}{q_t}  \text{ in } \sigma^{-1}(\mathfrak{A}) \text{, by Proposition \ref{WordSub} } \\
        &\iff q_s \xrightarrow[]{a} q_t \text{ in  }\tau^{-1}(\sigma^{-1}(\mathfrak{A})) \text{, by Proposition \ref{WordSub} again.}
\end{align*}

\end{proof}

\subsection{The problem of the purely substitutive walk}

We underline the following property of desubstitutions of $\omega$-automata:

\begin{fact}\label{fct:DesubFinite}
Let $\mathfrak{A}$ be an $\omega$-automaton,
let $\mathfrak{S}(\mathfrak{A})$ be the set of all $\omega$-automata which have the same alphabet, the same set of states and the same initial states as $\mathfrak{A}$. 
For any homomorphism $\sigma$ on $\mathcal{A}$, $\sigma^{-1}(\mathfrak{A})$ is an element of $\mathfrak{S}(\mathfrak{A})$.
\end{fact}

The crucial point is that $\mathfrak{S}(\mathfrak{A})$ is finite: given $\mathfrak{A} = (\mathcal{A}, Q,I,T)$, an element of $\mathfrak{S}(\mathfrak{A})$ is identified by its transitions, which form a subset of $(Q \times \mathcal{A} \times Q)$, so $\text{Card}(\mathfrak{S}(\mathfrak{A})) = 2^{|Q|^2 \times |\mathcal{A}|}$.
We could work on a subset of $\mathfrak{S}(\mathfrak{A})$ by identifying $\omega$-automata with the same language \cite{bassino_david_sportiello_2011}, but finiteness is sufficient for our results.

Given $\mathfrak{A}$ an $\omega$-automaton and $\sigma$ a homomorphism, $\sigma^{-1}$ defines a dynamic on the finite set $\mathfrak{S}(\mathfrak{A})$. 
By the pigeonhole principle:

\begin{fact}\label{fct:LoopAutomata}
Let $\mathfrak{A}$ be an $\omega$-automaton, and $\sigma$ be a homomorphism.
Then there exist $n < m  \leq |\mathfrak S(\mathfrak A)|+1$ such that $\sigma^{-n}(\mathfrak{A}) = \sigma^{-m}(\mathfrak{A})$.
\end{fact}

In the remainder of the section, we prove that, given an $\omega$-automaton $\mathfrak{A}$ and a substitution $\sigma$, the problems of finding a fixed point of $\sigma$ or a purely substitutive word generated by $\sigma$ accepted by $\mathfrak{A}$ are decidable.

A purely substitutive word generated by an erasing homomorphism $\sigma$ is also generated by a non-erasing homomorphism $\tau$ (that is, a substitution) that can be effectively constructed: remove every erased letter from $\mathcal{A}$ and from the images of $\sigma$, and repeat the process.
Thus, we assume $\sigma$ itself is a substitution.

\begin{proposition}\label{prop:FixedPointPower}
Let $\mathfrak{A}$ be an $\omega$-automaton, let $\sigma$ be a substitution and let $n < m \leq |\mathfrak S(\mathfrak A)|+1$ such that $\sigma^{-n}(\mathfrak{A}) = \sigma^{-m}(\mathfrak{A})$.
Then, $\mathfrak{A}$ accepts a fixed point for $\sigma^k$ for some $k \geq 1$ iff  $\mathcal{L}_\infty(\sigma^{-n}(\mathfrak{A}))$ is nonempty.
\end{proposition}

\begin{proof}
If $\mathcal{L}_\infty(\sigma^{-n}(\mathfrak{A}))$ is empty, then, by Propositions~\ref{prop:DesubAuto} and \ref{prop:DesubComp}, $\mathcal{L}_\infty(\sigma^{-p}(\mathfrak{A}))$ is empty for every $p \geq n$.
Let $k \geq 1$: if there were a fixed point $x$ for $\sigma^k$ accepted by $\mathfrak{A}$, we would have $x = \sigma^k(x) = \sigma^{kn}(x)$ by iterating.
So $x$ would be in $\mathcal{L}_\infty(\sigma^{-kn}(\mathfrak{A}))$ which is empty. By contradiction, there is no fixed point for any $\sigma^k$.

If $\mathcal{L}_\infty(\sigma^{-n}(\mathfrak{A}))$ is nonempty, let $x$ be a word accepted by $\sigma^{-n}(\mathfrak{A})$.
Again by Propositions~\ref{prop:DesubAuto} and \ref{prop:DesubComp}, because $\sigma^{-n}(\mathfrak{A}) = \sigma^{-m}(\mathfrak{A}) = \sigma^{-(m-n)}(\sigma^{-n}(\mathfrak{A}))$, $x$ is accepted by $\sigma^{-j(m-n)}(\sigma^{-n}(\mathfrak{A})) = \sigma^{-(n + j(m-n))}(\mathfrak{A})$ for all $j \in \mathbb{N}$.
This means that $\sigma^{n+j(m-n)}(x)$ is accepted by $\mathfrak{A}$ for all $j \in \mathbb{N}$.
Consider an adherence value $\Tilde{x}$ of the sequence $(\sigma^{n+j(m-n)}(x))_{j \in \mathbb{N}}$.
By compactness of the language of an $\omega$-automaton, $\Tilde{x} \in \mathcal{L}_{\infty}(\sigma^{-n}(\mathfrak{A}))$.

We define $Q_{\sigma^{m-n}} \subseteq \mathcal{A}$ the set of quiet letters for $\sigma^{m-n}$: $a \in Q_{\sigma^{m-n}}$ if $|\sigma^{j(m-n)}(a)| = 1$ for all $j \geq 1$. 
Let $k = \inf \{i \in \mathbb{N}\ |\ \Tilde{x}_i \notin Q_{\sigma^{m-n}}\}$ ($k$ may be infinite).
Then, for $i < k$, because $\sigma$ is a (nonerasing) substitution and every letter in $\Tilde{x}_{\llbracket 0,k-1\rrbracket}$ is quiet, $\sigma^{(m-n)}(\Tilde{x})_i = \sigma^{(m-n)}(\Tilde{x}_i)$.
In addition, because $\Tilde{x}$ is an adherence value of $(\sigma^{n+j(m-n)}(x))_{j \in \mathbb{N}}$, there is $r(\Tilde{x}_i) \geq 1$ such that $\sigma^{r(\Tilde{x}_i)\cdot(m-n)}(\Tilde{x}_i) = \Tilde{x}_i$ for every position $i < k$. 
Since $\mathcal A$ is finite, $(r(\Tilde{x}_i))_{0\leq i<k}$ contains only finitely many values, so we can define $r = \text{lcm} \{r(\Tilde{x}_i)\}$.

When $k<\infty$, there exists $q \geq 1$ such that $|\sigma^{q(m-n)}(\Tilde{x}_k)| > 1$ and $\Tilde{x}_k \sqsubseteq_p \sigma^{q(m-n)}(\Tilde{x}_k)$, for the same reason that $\Tilde{x}$ is an adherence value of $(\sigma^{n+j(m-n)}(x))_{j \in \mathbb{N}}$.
If $k = \infty$, we set $q = 1$.

Then, by concatenation, $\Tilde{x}_{\llbracket 0,k \rrbracket} \sqsubseteq_p \sigma^{rq(m-n)}(\Tilde{x}_{\llbracket 0,k \rrbracket})$.
Thus, $(\sigma^{jrq(m-n)}(\Tilde{x}))_{ j \in \mathbb{N}}$ has a limit, which is a fixed point for $\sigma^{rq(m-n)}$, and by compactness of $\mathcal{L}_{\infty}(\mathfrak{A})$, is accepted by $\mathfrak{A}$.
\end{proof}

Because the emptiness of the language of an $\omega$-automaton is decidable:

\begin{corollary}\label{cor:FixedPointPower}
The following problem is decidable:
\problem{An $\omega$-automaton $\mathfrak A$ and a substitution $\sigma$}
{Does $\mathfrak A$ accept a fixed point of $\sigma^k$ for some $k$?}
\end{corollary}

As is, this method alone cannot determine, for instance, whether $\mathfrak{A}$ accepts a fixed point for $\sigma$ itself (without power). This problem is still decidable, as we show later in Proposition~\ref{prop:FixedPoint} with a refinement of this method. In appendix, we provide examples where $\mathfrak{A}$ accepts fixed points for some $\sigma^k$ where $k$ does not correspond to $m-n$ where $n < m$ are the minimal powers such that $\sigma^{-n}(\mathfrak{A}) = \sigma^{-m}(\mathfrak{A})$.

Now, we come back to purely substitutive words.
A purely substitutive word generated by $\sigma$ is also a fixed point for some $\sigma^k$ (in fact, it is a fixed point for every $\sigma^j$ with $j \geq 1$).

\begin{proposition}\label{prop:PureSubsWord}
    Let $\mathfrak{A}$ be an $\omega$-automaton, $\sigma$ a substitution and $n < m \leq |\mathfrak S(\mathfrak A)|+1$ such that $\sigma^{-n}(\mathfrak{A}) = \sigma^{-m}(\mathfrak{A})$.
    Let $RP_\sigma \subseteq \mathcal{A}$ be the set of letters $b$ that are right-prolongable for $\sigma$, i.e. $b \sqsubseteq_p \sigma(b)$ and $b \neq \sigma(b)$.
    Then, $\mathfrak{A}$ accepts a purely substitutive word generated by $\sigma$ iff $\sigma^{-n}(\mathfrak{A})$ accepts an infinite word beginning with an element of $RP_\sigma$.
\end{proposition}

\begin{proof}
    If $\mathfrak{A}$ accepts a purely substitutive word $u$ generated by $\sigma$, $u = \lim_{j \rightarrow \infty} \sigma^j(b)$ begins by an element of $RP_\sigma$.
    Since $\sigma(u) = u$, $\sigma^n(u)$ is accepted by $\mathfrak{A}$ so $u$ is accepted by $\sigma^{-n}(\mathfrak{A})$.

    On the converse, suppose that $\sigma^{-n}(\mathfrak{A})$ accepts an infinite word beginning by $b \in RP_\sigma$.
    Then, $\sigma^{m-n}(b)$ labels an accepting computation on $\sigma^{-m}(\mathfrak{A}) = \sigma^{-n}(\mathfrak{A})$.
    By iteration, for every $k \geq 1$, we have that $\sigma^{k(m-n)}(b)$ labels an accepting computation on $\sigma^{-n}(\mathfrak{A})$, so $\sigma^{n + k(m-n)}(b)$ always labels an accepting computation on $\mathfrak{A}$. By compactness, $u = \lim\limits_{k \rightarrow \infty} \sigma^{n + k(m-n)}(b)$ is accepted by $\mathfrak{A}$.
    Now, because $b \in RP_\sigma$, the word $\lim\limits_{j \rightarrow \infty} \sigma^j(b)$ is defined and equal to $u$.
    Therefore $u$, the purely substitutive word generated by $\sigma$ on the letter $b$, is accepted by $\mathfrak{A}$.
\end{proof}

The following result already appeared in \cite{salo_2022}, but an erratum clarified that some cases were not covered \cite{salo_erratum_2022}. It is a parallel to a result in \cite{carton_thomas_2002}.
Our proof is essentially the same, but writing the proof through the lens of desubstitution makes it easier to extend the result to other decision problems.

\begin{corollary}
    The problem of the purely substitutive walk is decidable:
    \problem{an $\omega$-automaton $\mathfrak{A}$, a homomorphism $\sigma$.}{Does $\mathfrak{A}$ accept some purely substitutive word generated by $\sigma$?}
\end{corollary}

This result extends to morphic words: to find a morphic word generated by $\sigma$ and $\tau$ accepted by $\mathfrak{A}$, find a purely substitutive word generated by $\sigma$ accepted by $\tau^{-1}(\mathfrak{A})$.

We now extend the method used to prove Proposition~\ref{prop:PureSubsWord} to solve the question of finding a pure fixed point for a substitution $\sigma$ in an $\omega$-automaton. This improves Proposition~\ref{prop:FixedPointPower} where we found a fixed point for some power of $\sigma$.

\begin{proposition}\label{prop:FixedPoint}
    The problem of the fixed point walk is decidable:
    \problem{an $\omega$-automaton $\mathfrak{A}$, a substitution $\sigma$.}{Does $\mathfrak{A}$ accepts a fixed point for $\sigma$?}
\end{proposition}

\begin{proof}
    Let $x$ be a fixed point for $\sigma$ and define $FP_\sigma = \{b \in \mathcal{A}\ |\ \sigma(b) = b\}$ be the set of letters which are fixed points under $\sigma$.
    There are two cases:
\begin{enumerate}[topsep=0pt]
\item $x$ is an infinite word on the alphabet $FP_\sigma$.

\item there is a letter $a$ appearing in $x$ such that $\sigma(a) \neq a$.
    Suppose that $a$ is the first such letter in $x$. Then $x$ can be written as $x = p a x'$ where $p$ is a word on $FP_\sigma$.
    We have that $x = \sigma(x) = \sigma(p)\sigma(a) \sigma(x') = p \sigma(a) \sigma(x')$.
    So $a \sqsubseteq_p \sigma(a)$: $a$ is right-prolongable for $\sigma$, so $\lim\limits_{n \rightarrow \infty} \sigma^n(a)$ exists. Since $x = \sigma^n(x) = p \sigma^n(a) \sigma^n(x')$ for every $n \in \mathbb{N}$, by compactness, $x = p \lim\limits_{n \rightarrow \infty} \sigma^n(a)$.
\end{enumerate}
    The algorithm works as follows.
    First (case 1), check whether $\mathfrak{A}$ accepts a word on the alphabet $FP_\sigma$.
    Second (case 2), define a new automata $\mathfrak{A}'$ which is equal to $\mathfrak{A}$ except that the set of initial states is all the states reachable in $\mathfrak A$ by words in $FP_\sigma$, and check (by the previous algorithm) if $\mathfrak{A}'$ accepts a purely substitutive word generated by $\sigma$.

    The algorithm outputs "yes" if either case is satisfied, and "no" otherwise.
\end{proof}

\subsection{The problem of the infinitely desubstitutable walk}

In this section, we suppose that $\mathfrak A$ is an $\omega$-automaton and $\mathcal{S}$ is a finite set of substitutions (i.e. nonerasing homomorphisms, as is usual when studying multiple homomorphisms) on a single alphabet $\mathcal{A}$. We prove that the problem of finding an infinitely desubstitutable (infinite) word accepted by $\mathfrak A$ is decidable. To study this question, we introduce a meta-$\omega$-automaton:
each symbol is a substitution, and each state is an $\omega$-automaton.

\begin{definition}[The meta-$\omega$-automaton $\mathcal{S}^{-\infty}(\mathfrak{A})$]
We define the $\omega$-automaton $\mathcal{S}^{-\infty}(\mathfrak{A}) = (\mathcal{S}, D(\mathfrak{A}), \{\mathfrak{A}\}, \mathcal{T})$ with the alphabet $\mathcal{S}$, the set of states $D(\mathfrak{A}) = \{\sigma^{-1}(\mathfrak{A}),\ \sigma \in \mathcal{S}^\ast\}$, $\mathfrak A$ the only initial state and set of transitions $\mathcal{T} = \{ \mathfrak{B} \xrightarrow[]{\sigma} \sigma^{-1}(\mathfrak{B})\ |\ \mathfrak{B} \in D(\mathfrak{A}), \sigma \in \mathcal{S}\}$.
\end{definition}

Because $D(\mathfrak{A}) \subseteq \mathfrak{S}(\mathfrak{A})$ is finite (see Fact~\ref{fct:DesubFinite}), $\mathcal{S}^{-\infty}(\mathfrak{A})$ is computable.
We prove that directive sequences of words accepted by $\mathfrak{A}$ correspond to \emph{non-nilpotent} walks in $\mathcal{S}^{-\infty}(\mathfrak{A})$, that is, walks $(\mathfrak{B}_n)_{n \in \mathbb{N}}$ such that $\mathcal{L}_\infty(\mathfrak{B}_n) \neq \emptyset$ for all $n$.

\begin{proposition}
\label{prop:InfDesub}
There exists $x$ an infinite word infinitely desubstitutable by $(\sigma_n)_{n \in \mathbb N}$ accepted by $\mathfrak{A}$ if, and only if, there is a non-nilpotent infinite walk in $\mathcal{S}^{-\infty}(\mathfrak{A})$ labeled by $(\sigma_n)_{n \in \mathbb{N}}$.\end{proposition}


\begin{corollary}\label{prop:SAdicLanguage}
    The set of directive sequences of infinitely desubstitutable words accepted by $\mathfrak A$ is the language of some $\omega$-automaton.
\end{corollary}

\begin{proof}[of Proposition~\ref{prop:InfDesub}]
First, let $x$ be an infinitely desubstitutable word with directive sequence $(\sigma_n)_{n \in \mathbb N}$, and let $(x_n)_{n \in \mathbb N}$ be the sequence of desubstituted words.
Then, by Proposition~\ref{prop:DesubAuto}, $x_n \in \mathcal{L}_\infty((\sigma_1 \circ \dots \circ \sigma_{n-1})^{-1}(\mathfrak{A}))$.
So the walk $(\sigma_{\llbracket 0, n \rrbracket}^{-1}(\mathfrak{A}))_{n \in \mathbb N}$ is non-nilpotent and labeled by $(\sigma_n)_{n \in \mathbb N}$.

Second, let $(\sigma_n)_{n \in \mathbb{N}}$ label a non-nilpotent infinite walk in $\mathcal{S}^{-\infty}(\mathfrak{A})$.
It means that each language $(\sigma_1 \circ \dots \circ \sigma_k)^{-1}(\mathcal{L}_\infty(\mathfrak{A}))$ is nonempty.
Now, consider the sequence $((\sigma_1 \circ \dots \circ \sigma_n)(\mathcal{L}_\infty((\sigma_1 \circ \dots \circ \sigma_n)^{-1}(\mathfrak{A}))))_{n \in \mathbb{N}}$.
It satisfies the following:

\begin{enumerate}[topsep = 0pt]
    \item each element of the sequence is included in $\mathcal{L}_\infty(\mathfrak{A})$;
    \item because $\mathcal{L}_\infty((\sigma_1 \circ \dots \circ \sigma_n)^{-1}(\mathfrak{A})))$ is compact and nonempty, and $(\sigma_1 \circ \dots \circ \sigma_n)$ is continuous, every element of the sequence is compact and nonempty;
    \item the sequence is decreasing for inclusion.
\end{enumerate}

By Cantor's intersection theorem, there is a point $x$ in the intersection of every element of the sequence.
This point $x$ is desubstitutable by any $\sigma_1 \circ \dots \circ \sigma_k$, thus it is infinitely desubstitutable by the sequence $(\sigma_n)_{n \in \mathbb{N}}$.
\end{proof}

With Proposition \ref{prop:InfDesub}, we can deduce the decidability of the existence of an infinitely desubstitutable word accepted by an $\omega$-automaton $\mathfrak A$.
First, build $\mathcal{S}^{-\infty}(\mathfrak{A})$; second, remove the states corresponding to $\omega$-automata with an empty language; last, check whether there is an infinite walk.

\begin{proposition}\label{prop:InfDesubWalk}
The problem of the infinitely desubstitutable walk is decidable:
\problem{a finite set of substitutions $\mathcal{S}$, an $\omega$-automaton $\mathfrak{A}$}{does $\mathcal{L}_\infty(\mathfrak{A})$ contain a word which is infinitely desubstitutable by $\mathcal{S}$?}
\end{proposition}

\subsection{The problem of the Büchi infinitely desubstitutable walk}

Proposition~\ref{prop:InfDesubWalk} does not apply directly to Sturmian words. Indeed, the classical characterization of Sturmian words restricts the possible directive sequences. 

$\mathcal{S}_{St}$ is the set containing the four following substitutions, called (elementary) Sturmian morphisms, as described by \cite{berstel_seebold_2002}.
\vspace{-0.4em}
\[L_0 : \left\{\begin{array}{ccc}
    0 & \mapsto & 0 \\
    1 & \mapsto & 01
\end{array}\right.,\quad L_1 : \left\{\begin{array}{ccc}
    0 & \mapsto & 10 \\
    1 & \mapsto & 1
\end{array}\right.,\quad R_0 : \left\{\begin{array}{ccc}
    0 & \mapsto & 0 \\
    1 & \mapsto & 10
\end{array}\right.,\quad R_1 : \left\{\begin{array}{ccc}
    0 & \mapsto & 01 \\
    1 & \mapsto & 1
\end{array}\right.\]

\begin{theorem}[\cite{arnoux_fogg_2002}]
     \label{thm:SturmArnoux}
    A word is Sturmian iff it is infinitely desubstitutable by a directive sequence $(\sigma_n)_{n \in \mathbb{N}} \subset \mathcal{S}_{St}$ that alternates infinitely in type, i.e.: $\nexists N \in \mathbb{N}, (\forall n \geq N, \sigma_n \in \{L_0, R_0\})$ or $(\forall n \geq N, \sigma_n \in \{L_1, R_1\})$.
\end{theorem}

This characterization is usually expressed in the $S$-adic framework, but is equivalent in this context \cite{richomme_2021}. 
In this section, we generalize Proposition~\ref{prop:InfDesubWalk} to Sturmian words and more general restrictions on the directive sequence.

\begin{proposition}
    \label{prop:SturmianWalk}
    The problem of the Sturmian walk is decidable:
   \problem{an $\omega$-automaton $\mathfrak{A}$.}{is there a Sturmian infinite word accepted by $\mathfrak{A}$?}
\end{proposition}

\begin{proof}
    Consider the associated representation automaton $\mathcal{S}_{St}^{-\infty}(\mathfrak{A})$.
    According to Proposition~\ref{prop:InfDesub} combined with Theorem~\ref{thm:SturmArnoux},  there is a Sturmian infinite word accepted by $\mathfrak{A}$ if, and only if, there is an infinite computation accepted by $\mathcal{S}_{St}^{-\infty}(\mathfrak{A})$ labeled by a word $(\sigma_n)_{n \in \mathbb{N}}$ which alternates infinitely in type.
    This last condition is decidable: compute the strong connected components of $\mathfrak{A}$, and check that there is at least one strongly connected component $C$ which contains two edges labeled by substitutions in $\{L_0, R_0\}$ and $\{L_1, R_1\}$, respectively.    
\end{proof}

In this case, the condition of alternating infinitely in type is easy to check: it can actually be described using a Büchi $\omega$-automaton on the alphabet $\mathcal S$. Proposition~\ref{prop:SturmianWalk} generalizes to every such condition.

\begin{definition}
    Let $\mathcal{S}$ be a set of substitutions, and $\mathfrak{R}$ a Büchi $\omega$-automaton on the alphabet $\mathcal{S}$.
    Define $X_{\mathfrak{R}}$ as $\{x \in \mathcal{A}^\mathbb{N}\ |\ \exists (\sigma_n)_{n \in \mathbb{N}} \in \mathcal{L}_\infty(\mathfrak{R}), x \text{ is inf. desub. by } (\sigma_n) \}$.
\end{definition}

\begin{proposition}
    \label{prop:Generalisation1}
    The following problem is decidable:
    \problem{an $\omega$-automaton $\mathfrak{A}$, a finite set of substitutions $\mathcal{S}$, a Büchi $\omega$-automaton $\mathfrak{R}$ on the alphabet $\mathcal{S}$}{is there an infinite word of $X_\mathfrak{R}$ accepted by $\mathfrak{A}$?}
\end{proposition}

\begin{proof}
    The question of the problem is equivalent to: is $\mathcal{L}_\infty(\mathfrak{R}) \cap \mathcal{L}_\infty(\mathcal{S}^{-\infty}(\mathfrak{A})) \neq \emptyset$?
    The intersection between a Büchi $\omega$-automaton and an $\omega$-automaton is a Büchi $\omega$-automaton that can be effectively constructed \cite{perrin_pin_2004}, and checking the non-emptiness of a Büchi $\omega$-automaton is decidable.
\end{proof}

The interest of Proposition~\ref{prop:Generalisation1} is that there exists a zoology of families of words which have a characterization by infinite desubstitution.
For instance, Proposition~\ref{prop:Generalisation1} applies to Arnoux-Rauzy words \cite{arnoux_rauzy_1991} and to minimal dendric ternary words \cite{gheeraert_lejeune_leroy_2021}.
We also characterize the set of allowed directive sequences akin to Corollary~\ref{prop:SAdicLanguage}: the set of directive sequences on $\mathcal{S}$ accepted by the Büchi $\omega$-automaton $\mathfrak R$ that define a word accepted by $\mathfrak A$ is itself recognized by a Büchi $\omega$-automaton.

Let us translate Proposition~\ref{prop:Generalisation1} in more dynamical terms:

\begin{proposition}
The following problem is decidable:

\problem{a set of substitutions $\mathcal S$, a Büchi $\omega$-automaton $\mathfrak{R}$ on the alphabet $\mathcal S$ and a sofic shift $\mathbb{S}$.}
{Is $\mathbb{S} \cap X_\mathfrak{R}$ empty?}
\end{proposition}

\subsection{Application to the coding of Sturmian words}

Here is an example of a natural question from combinatorics on words that we solve on Sturmian words, even though the method generalizes easily.
Let $W$ be a finite set of finite words on $\{0,1\}$.
Consider $W^\omega$ the set of infinite concatenations of elements of $W$, i.e. $W^\omega = \{x \in \{0,1\}^\mathbb{N}\ |\ \exists (w_n)_{n \in \mathbb{N}} \subseteq W, x = \lim\limits_{n \rightarrow \infty} w_0 w_1 \dots w_n\}$.

\begin{proposition}
    The following problem is decidable:
    \problem{$W$ a finite set of words on \{0,1\}}{does $W^\omega$ contain a Sturmian word?}
\end{proposition}

\begin{proof}
    The language $W^\omega$ is $\omega$-regular: there is an $\omega$-automaton $\mathfrak{A}_W$ such that $\mathcal{L}_\infty(\mathfrak{A}_W) = W^\omega$.
    Then, $W^\omega$ contains a Sturmian word iff $\mathfrak{A}_W$ accepts a Sturmian word, which is decidable by Proposition~\ref{prop:SturmianWalk}.
\end{proof}

\section{About $\omega$-automata recognizing Sturmian words}

In this Section, we focus on Sturmian words and show that the language of Sturmian words is as far as possible from being regular, in the sense that an $\omega$-automaton may only accept a Sturmian word if it accepts the image of the full shift under a Sturmian morphism.

\begin{theorem}
\label{prop:SturmianOmega}
    Let $\mathcal{S} = \mathcal{S}_{St}$ be the set of elementary Sturmian morphisms as defined earlier, and let $\mathfrak{A}$ be an $\omega$-automaton.
    If $\mathfrak{A}$ accepts a Sturmian word, then $\exists \sigma \in \mathcal{S}_{St}^\ast, \sigma(\mathcal{A}^\mathbb{N}) \subseteq \mathcal{L}_\infty(\mathfrak{A})$.
\end{theorem}

This is equivalent to the presence of a total automaton in $\mathcal{S}^{- \infty}(\mathfrak A)$: an $\omega$-automaton $\mathfrak A$ is total if $\mathcal{L}_\infty(\mathfrak A) = \mathcal{A}^\mathbb{N}$.
Totality is a stable property under any desubstitution.

To prove Theorem~\ref{prop:SturmianOmega}, we introduce the following technical tools.

\begin{definition}\label{def:PropH}
    Let $\mathfrak{A}$ be an $\omega$-automaton  on $\mathcal{A} = \{0,1\}$.
    A state $q$ of $\mathfrak A$ has property $(H)$ if $(\exists q_s, q \xrightarrow[]{0} \transo{q_s}{}{\cdots} \in \mathfrak A) \Leftrightarrow (\exists q_t, q \xrightarrow[]{1} \transo{q_t}{}{\cdots} \in \mathfrak A)$, where $\transo{q_s}{}{\cdots}$ means that there is an infinite computation starting from $q_t$ in $\mathfrak{A}$.
\end{definition}

If all states of $\mathfrak A$ have property (H), there are two possibilities: if there is no infinite computation starting on an initial state, the infinite language of $\mathfrak{A}$ is empty; otherwise, $\mathfrak{A}$ is total.

    \begin{lemma}
    \label{lem:TechnicalLemma}
      Let $\mathfrak C$ be an $\omega$-automaton, and $\phi \in \mathcal{S}_{St}^\ast$ starting with $L_0$ and ending with $L_1$ such that $\phi^{-1}(\mathfrak{C}) = \mathfrak{C}$. Then, every state of $\mathfrak{C}$ has property $(H)$.
    \end{lemma}

    \begin{proof}[of Lemma \ref{lem:TechnicalLemma}]
        Let $\mathfrak C = (\{0,1\}, Q_\mathfrak{C}, I_\mathfrak{C}, T_\mathfrak{C})$, and $q \in Q_\mathfrak{C}$.
        First, suppose that $q \xrightarrow[]{0} \transo{q_t}{}{\cdots}$ is a computation in $\mathfrak C$.
        Then $q \xrightarrow[]{0} q_t$ is also a transition of $\phi^{-1}(\mathfrak{C})$.
        So $\trans{q}{\phi(0)}{q_t}$ is a computation in $\mathfrak{C}$.
        Because $\phi$ ends with $L_1$, $\phi(1) \sqsubseteq_p \phi(0)$.
        So $\trans{q}{\phi(1)}{q_u} \trans{}{m}{q_t} \transo{}{}{\cdots}$ is a computation in $\mathfrak C$, with some $q_u \in Q_\mathfrak{C}$ and $\phi(0) = \phi(1)m$.
        Now, using $\mathfrak C = \phi^{-1}(\mathfrak C)$, $q \xrightarrow[]{1} q_u \trans{}{m}{q_t} \transo{}{}{\cdots}$ is a computation in $\mathfrak C$.

        Conversely, if $q \xrightarrow[]{1} \transo{q_t}{}{\cdots}$ is a computation in $\mathfrak C = \phi^{-1}(\mathfrak C)$, there is also $\trans{q}{\phi(1)}{} \transo{q_t}{}{\cdots}$
        Because $\phi$ begins with $L_0$, $\phi(1) = 0m$ for some finite $m$.
        So the last computation can be written $q \xrightarrow[]{0} \trans{q_u}{m}{q_t} \transo{}{}{\cdots}$
    \end{proof}

\begin{proof}[of Theorem~\ref{prop:SturmianOmega}]
    Let $x$ be a Sturmian word accepted by $\mathfrak{A}$.
    Consider the transformation of $\omega$-automata $\text{forget} : (\mathcal{A}, Q,I,T) \mapsto (\mathcal{A}, Q,Q,T)$ which makes all states initial.
    Then, $\text{forget}(\mathfrak{A})$ also accepts $x$, and $\mathcal{L}_\infty(\text{forget}(\mathfrak A))$ is a sofic shift.
    Then $\overline{\bigcup\limits_{n \geq 0} S^n(x)}$, which is the orbit of $x$ under the shift $S$, is contained in $\mathcal{L}_\infty(\text{forget}(\mathfrak{A}))$.
    Let $\chi(x)$ be the Sturmian characteristic word associated with  $x$ (see \cite{perrin_restivo_2012}): it belongs to the orbit of $x$, so it is accepted by $\text{forget}(\mathfrak{A})$.
    Then, $\chi(x) = \lim\limits_{n \rightarrow \infty} \sigma_0 \circ \dots \circ \sigma_n(a_n)$ with $(\sigma_n)_{n \in \mathbb{N}} \subseteq \mathcal{S}_{St}$ a sequence that alternates infinitely in type (see Theorem~\ref{thm:SturmArnoux}).
    Besides, because $\chi(x)$ is a characteristic word, it represents the orbit of zero from the point of view of circle rotation (see \cite{perrin_restivo_2012}): when combined with Proposition~2.7 of \cite{berthe_holton_zamboni_2006}, it yields that
    $(\sigma_n)_{n \in \mathbb{N}} \subseteq \{L_0, L_1\}^\mathbb{N}$.
    By the pigeonhole principle, there is an $\omega$-automaton $\mathfrak{B}$ that appears infinitely often in the sequence $(\sigma_{\llbracket0,n\rrbracket}^{-1}(\text{forget}(\mathfrak{A})))_{n \in \mathbb{N}} \subseteq \mathfrak{S}(\text{forget}(\mathfrak{A}))$.
    Thus, we can find a substitution $\tau$ such that $\mathfrak{B} = \tau^{-1}(\mathfrak{B})$ and $\tau \in \{L_0,L_1\}^\ast \setminus (L_0^\ast \cup L_1^\ast)$.
    Because $\tau$ contains both $L_0$ and $L_1$, there are two cases:
        \begin{enumerate}[topsep = 0pt]
            \item $L_1 L_0 \sqsubseteq_f \tau$: we can write $\tau = p_\tau L_1 L_0 s_\tau$.
            Let $\mathfrak{B}' = (p_{\tau} \circ L_1)^{-1}(\mathfrak{B})$ and $\tau' = L_0 \circ s_\tau \circ p_\tau \circ L_1$: we have that $\tau'^{-1}(\mathfrak{B}') = \mathfrak{B}'$.
            \item $L_1 L_0 \not\sqsubseteq_f \tau$: then, $\tau$ begins with a $L_0$ and ends with a $L_1$.
        \end{enumerate}
    In both cases, we can come back to the case where $\tau$ begins with a $L_0$ and ends with a $L_1$.

    Now, we apply Lemma~\ref{lem:TechnicalLemma} to show that every state of $\mathfrak{B}$  has property $(H)$. 
    $\mathfrak{B}$ can be written as $\psi^{-1}(\text{forget}(\mathfrak{A}))$ for some Sturmian morphism $\psi$.
    Since the transformation $\text{forget}$ does not modify the transitions of an $\omega$-automaton, this yields that every state of $\psi^{-1}(\mathfrak{A})$ also has property $(H)$. Since by assumption $\psi^{-1}(\mathfrak{A})$ accepts an infinite word, it follows that it is total.
\end{proof}

Let $f$ be the Fibonacci word, i.e. the substitutive word associated with the substitution $\sigma_f(0) = 01, \sigma_f(1) = 0$.
Since Lemma~\ref{lem:TechnicalLemma} holds when $\phi = \sigma_f^n$ ($n \geq 1$), by adapting the proof of Theorem~\ref{prop:SturmianOmega}, we obtain an equivalent statement for $f$:

\begin{corollary}\label{CompleteFibo}
    Let $\mathfrak{A}$ be an $\omega$-automaton which accepts $f$.
    Then, there exists $n \in \mathbb{N}$ such that $\sigma_f^{-n}(\mathfrak{A})$ is total.
\end{corollary}

This combinatorial result can be thought in dynamical terms:

\begin{corollary}
    A sofic subshift $\mathbb S$ contains $f$ iff $\mathbb S$ contains some $\sigma_f^n(\mathcal{A}^\mathbb{N})$.
\end{corollary}

Because the Fibonacci word is aperiodic, containing $f$ means that there is a substitution $\tau$ such that $\tau(\mathcal{A}^\mathbb{N})$ is contained in $\mathbb{S}$.
Because the Fibonacci word is Sturmian, Berstel and Séébold \cite{berstel_seebold_2002} established that $\tau$ had to be a Sturmian morphism.
This new analysis specifies that $\tau$ can be chosen a power of $\sigma_f$.

\section{Open questions}

\begin{itemize}
    \item Following Proposition~\ref{prop:InfDesubWalk}, find an algorithm to find an accepted $\mathcal{S}$-adic word. There are technical difficulties to take into account the growth of the directive sequence, which should be solvable using results from \cite{richomme_2021}.
    \item Can our methods extend to Büchi $\omega$-automata, as in \cite{carton_thomas_2002}?
The difficulty is that the language of Büchi $\omega$-automata is not always compact, so Proposition~\ref{prop:FixedPointPower} does not apply. It may be possible to extend methods from \cite{carton_thomas_2002}.
\item For which sets of substitutions does Theorem~\ref{prop:SturmianOmega} hold?
\end{itemize}

\bibliographystyle{splncs04}
\bibliography{bibliography}

\newpage

\section{Appendix}

\subsection{Counterexamples for Proposition \ref{prop:FixedPointPower}}

Using the notation of Proposition \ref{prop:FixedPointPower}, consider $n < m$ minimal such that $\sigma^{-n}(\mathfrak{A}) = \sigma^{-m}(\mathfrak{A})$. There is no clear relationship between $n$, $m$ and $k$ the power of the fixed point accepted by $\mathfrak A$.

Here is a example with $m-n \geq 2$, but $\mathfrak A$ accepts a fixed point for $\sigma$.

\vspace{-1.5em}

\begin{figure}[ht!]
    \hfill    
    
        \begin{subfigure}[b]{0.52\textwidth}
    \centering
    $\sigma : \left\{\begin{array}{ccc}
    0 & \mapsto & 0 \\
    1 & \mapsto & 2 \\
    2 & \mapsto & 1 \\
    \end{array}\right.$
    \vspace{0.1\textwidth}
    \end{subfigure}   
    \begin{subfigure}[b]{0.4\textwidth}
         \centering
          \begin{tikzpicture}

\node[circle, draw] (a) at (0,0) {};
\node[circle, draw] (b) at (1,0) {};
\node[circle, draw] (c) at (0.5, 0.866) {};

\draw[->, >=latex, loop left] (a) to node[left]{0} (a);
\draw[->, >=latex, loop above] (c) to node[above]{0} (c);
\draw[->, >=latex, loop right] (b) to node[right]{0} (b);
\draw[->, >=latex] (a) to node[below]{1} (b);
\draw[->, >=latex] (b) to node[right]{1} (c);
\draw[->, >=latex] (c) to node[left]{1} (a);

\end{tikzpicture}
          \label{CEx1}
    \end{subfigure}
    \hfill
    \hspace{0.05\textwidth}
    \caption{$n = 0$ and $m = 2$, but $\mathfrak A$ accepts $0^\infty$, which is a fixed point for $\sigma$.}
\end{figure}
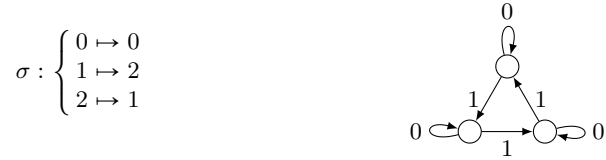

\vspace{-1.5em}

Next is an example where $m-n$ is lesser than the power required to have a fixed point:
\vspace{-1.5em}
\begin{figure}
    \centering
    \hfill
    \begin{subfigure}[b]{0.4\textwidth}
        $
    \arraycolsep=2pt\def\arraystretch{1}
    \tau : \left\{\begin{array}{ccc}
    0 & \mapsto & 11 \\
    1 & \mapsto & 00
    \end{array}\right.$ 
    \vspace{0.1\textwidth}
    \end{subfigure}
    \begin{subfigure}[b]{0.4\textwidth}
        \includegraphics[width = \textwidth]{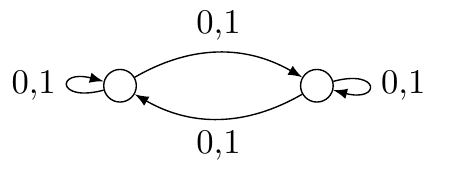}
        \centering
    \end{subfigure}
    \hfill
    \caption{$\tau^{-1}(\mathfrak{B}) = \mathfrak{B}$, so $m = 1$ and $n=0$.
$\tau$ has no fixed point: $\mathfrak{B}$ cannot contain a fixed point for $\tau^{m-n}$.
However, $\mathfrak{B}$ contains a fixed point for $\tau^2$.}
\end{figure}
\vspace{-2.5em}

\subsection{There is not always a total automaton in $\mathcal S^{-\infty}(\mathfrak A)$}
Theorem~\ref{prop:SturmianOmega} does not generalize straightforwardly to any set of substitutions: in general, $\mathcal S^{-\infty}(\mathfrak A)$ may not contain a total automaton, even under classical dynamical constraints.
For instance, consider the following $\omega$-automaton $\mathfrak{A}_{H}$ and susbtitution $\sigma_{H}$:

\vspace{-1.5em}

\begin{figure}
    \centering
    \hfill
    \begin{subfigure}[b]{0.4\textwidth}
    \centering
    $\sigma_{H} : \left\{ \begin{array}{ccc}
    0 & \mapsto & 0120 \\
    1 & \mapsto & 11220011 \\
    2 & \mapsto & 222000111222
\end{array} \right.$
\vspace{.1\textwidth}
     \end{subfigure}
    \hfill
   \begin{subfigure}[b]{0.4\textwidth}
    \centering
    \begin{tikzpicture}
    \node[circle, draw] (a) at (0,0) {};
    \node[circle, draw] (b) at (1,0) {};
    \node[circle, draw] (c) at (0.5, 0.866) {};

    \draw[->, >=latex, loop left] (a) to node[left]{0} (a); 
    \draw[->, >=latex] (b) to node[below]{0} (a); 
    \draw[->, >=latex, loop above] (c) to node[left]{1} (c); 
    \draw[->, >=latex] (a) to node[left]{1} (c); 
    \draw[->, >=latex, loop right] (b) to node[right]{2} (b); 
    \draw[->, >=latex] (c) to node[right]{2} (b); 
\end{tikzpicture}
     \end{subfigure}
    \hfill
    \caption{An $\omega$-automaton stable by desubstitution by $\sigma_H$.}
\end{figure}
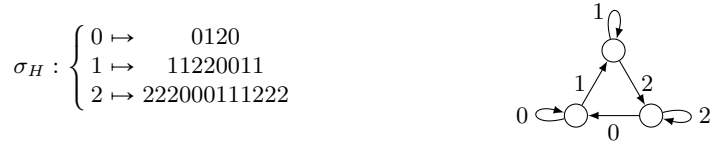

\vspace{-1.5em}

Notice that $\sigma_{H}$ is primitive, and that the three purely substitutive words generated by $\sigma_{H}$ are not eventually periodic.
However, $\mathfrak{A}_{H}$ is not total, and $\sigma_{H}^{-1}(\mathfrak{A}_{H}) = \mathfrak{A}_{H}$, so there is no total automaton in $\mathcal S^{-\infty}(\mathfrak A)$.
In dynamical terms, it means that the sofic shift contains the associated purely substitutive words, but contains no factor of the form $\sigma_{H}^k(\mathcal{A}^\mathbb{N})$.


\end{document}